\documentclass[letterpaper,twoside,11pt]{article}

\input epsf 

\usepackage{jair}
\usepackage{theapa}
\usepackage{algorithm2e}
\usepackage{graphicx}
\usepackage{booktabs}

\usepackage{lmodern}
\usepackage[T1]{fontenc}
\usepackage{textcomp}

\usepackage[cmex10]{amsmath}
\usepackage{amsfonts}
\usepackage{amssymb}
\usepackage{amsthm}

\newcommand{\eat}[1]{}

\theoremstyle{definition} \newtheorem{theorem}{Theorem}[section]
\theoremstyle{definition} \newtheorem{definition}[theorem]{Definition}
\theoremstyle{definition} 
\theoremstyle{definition} 
\theoremstyle{definition} 
\theoremstyle{definition} 
\theoremstyle{definition} \newtheorem{example}[theorem]{Example}
\theoremstyle{definition} 
\theoremstyle{definition} 
\theoremstyle{definition} 
\firstpageno{1}

\begin{document}

\title{Fairness in Combinatorial Auctions}
\author{\name Sumanth Sudeendra \email sumanth.s@iiitb.net \\
       \addr International Institute of Information Technology - Bangalore,\\
                      Bangalore, India.\\
       \AND
       \name Megha Saini \email msaini1@iit.edu\\
       \addr Illinois Institute of Technology,\\ 
        Chicago, USA.\\
       \AND      
       \name Shrisha Rao \email srao@iiitb.ac.in \\
       \addr International Institute of Information Technology - Bangalore,\\
                      Bangalore, India.\\
}
\maketitle

\begin{abstract}
  The market economy deals with many interacting agents such as buyers
  and sellers who are autonomous intelligent agents pursuing their own
  interests.  One such multi-agent system (MAS) that plays an
  important role in auctions is the combinatorial auctioning system
  (CAS).  We use this framework to define our concept of fairness in
  terms of what we call as ``basic fairness'' and ``extended
  fairness.''  The assumptions of quasilinear preferences and dominant
  strategies are taken into consideration while explaining fairness.
  We give an algorithm to ensure fairness in a CAS using a Generalized
  Vickrey Auction (GVA).  We use an algorithm of Sandholm to achieve
  optimality.  Basic and extended fairness are then analyzed according
  to the dominant strategy solution concept.
\end{abstract}


\section{Introduction} \label{intro}

The term ``auction'' refers to a mechanism of allocating single or
multiple resources to one or more agents (or bidders)~\cite{23}.  In
recent years, computer scientists, rather than just economists, are
interested in auctions.  The increase in computing power and improved
algorithms have paved the way for combinatorial auctions.  Here
multiple items are for sale by the auctioneer and bidders can bid for
a bundle of items (also called packages).  In a multi-agent system
(MAS), we consider these bidders and the auctioneer as autonomous
agents who act in a self-interested manner in their dealings with one
another.  Similarly, even in MAS dealing with resource allocation
other than by auction, there are self-interested autonomous
agents~\shortcite{4,19}.  We study a framework where optimality is a
desirable property but fairness is a required property.  An excellent
example of such a framework is a combinatorial auctioning system (CAS)
where the two most important issues pertaining to resource allocation
are optimality and fairness.  A CAS is a kind of MAS whereby the
bidders can express preferences over combinations of
items~\cite{15,13}.

We assume in this paper that an agent's valuation of an item does not
change based on other agents' private information (i.e., some evidence
which affects the valuation of an agent), that utilities are
quasilinear (i.e., utility is linear in terms of money), and that
there are no externalities (i.e., an agent that does not win an item
neither cares which other agent wins it, nor worries about how much
other agents pay for it)~\cite{23}.  This is realistic, as seen for
example in relation to the Nigerian Communications Commission auction
described below, where, ``The decision to charge bidders what they bid
was accepted by bidders and observers as fair and transparent despite
the difference in some of the payments for identical
licenses.''~\cite[p. 30]{29}  In such scenarios, each agent holds
different preferences over the various possible allocations and hence
concepts like individual rationality, fairness, optimality,
efficiency, etc., are important~\shortcite{7}.

We introduce the concept of fairness in the auction mechanism.
Although the notion of fairness is of course well known in general, it
does not seem to have been clearly defined with respect to auction
processes in particular.  We propose two types of fairness, namely
\emph{basic fairness} and \emph{extended fairness}.  We explain basic
fairness using the concept of equitable distribution along with the
respective preferences.  Extended fairness is explained such that
envy-freeness prevails in the allocation and the entire resource is
allocated to the winning bidder.  We introduce a \emph{fairness table}
consisting of fair values as perceived by bidders and auctioneer; this
is sealed at the start of the bidding process.  We give emphasis to
fairness, unlike the classical approach where revenue maximization is
the only goal required in auctions.  To achieve fairness, the proposed
algorithm explains a novel payment scheme which is applied at the end
of the bidding process, where we determine the final amount payable to
the auctioneer by the winning bidder.  We ensure that this process is
considered to be fair by both bidders and the auctioneer by means of
extended fairness.  We handle the special case of a tie in the bidding
process using equitable distribution, and ensure that basic fairness
is achieved.  The mathematical formulations of fairness concepts in
combinatorial auctions are explained, and a detailed analysis is
presented to highlight some of the properties exhibited by our payment
scheme.

In our mechanism, there are self-interested bidders and an auctioneer,
who express their perceptions of the fair value of the resources
through a data structure called the fairness table.  Here an
auctioneer acts as a facilitator to ensure that an item achieves its
fair value.  We consider optimality as the desired property and
fairness as the required property.  We illustrate this using a
combinatorial auction framework, in which multiple items are
simultaneously up for sale and a bidder can bid for any bundle of
items.  The optimal allocation of resources is discussed using an
algorithm of Sandholm where we obtain the winning bidders.  The
incentives to the winning bidders are provided through the Generalized
Vickrey Auction (GVA) and Algorithm~\ref{Algorithm1}. We apply the
fairness concept using the fairness table and
Algorithm~\ref{Algorithm1}.

The auctioning of the electromagnetic spectrum is one of the
well-known applications of combinatorial auctions.  The first-ever
combinatorial auctioning of the radio spectrum was held in Nigeria in
2002~\cite{29}.  A single-round sealed-bid combinatorial auction (not
the Simultaneous Multiple Round Auctions (SMRA) used in other
countries) was conducted for regional fixed wireless access (FWA)
licenses; the decision for this format was made as SMRA was
impractical in Nigeria due to its insufficient communication
infrastructure, and as the NCC did not wish to have a lengthy auction.
Some 67 out of the 80 licenses available were allocated, with
successful bids amounting to 3.78 billion naira (38 million US
dollars).  Here, the complimentarity and substitutability of licenses
were the important factors for choosing a combinatorial auction.  The
cost of the allocation process was an important factor, and the
Nigerian Communications Commission (NCC) did not want to discourage
smaller bidders, with its primary goal being efficiency and
transparency.

We find in the case of the combinatorial auction conducted by NCC the
following problem: ``The final choice of auction design rested heavily
on information revealed about the regional structure of demand from
initial applications.  It was therefore critical that the application
process created incentives for bidders to reveal such
information.''~\cite[p. 24]{29} This problem is taken care of in our
approach as we introduce the fairness table which is to be populated
at the beginning of the auctioning process.  This table allows us to
see the regional structure of demands through the fair values assigned
to the resources.  We provide higher rewards for bidders who
truthfully give their fair values, by Theorems~\ref{Thm1}
and~\ref{Thm2}.  

It was also observed by the NCC that bidders defaulted on their
winning bids in a significant number of cases, though not enough to
undermine the overall process~\cite{29}.  We can decrease the number
of bidders defaulting provided we satisfy Theorem~\ref{Thm2}.  This
makes it less rewarding for bidders to bid way beyond their capacity,
which in turn decreses the possibility of winning bidders defaulting.

Since the importance here is on transparent and fair allocation, we
can apply our method to ensure fairness in combinatorial auctions.  We
start with introducing some of the related work in
Section~\ref{related}.  Next, we explain different notions of fairness
with formal definitions in Section~\ref{definitions}.  This is
followed by our study of CAS in Section~\ref{fairCAS};
Algorithm~\ref{Algorithm1} in Section~\ref{algo} is used to extend the
payment scheme to achieve fairness in CAS with an example.
Section~\ref{analysis} gives a detailed analysis of fairness using
mechanism design under quasilinear settings~\cite{22}.  We conclude
with Section~\ref{conclusion}, which offers some conclusions about our
efforts, and some suggestions for further work along these lines.

\section{Related Work} \label{related}

In this section, we review different definitions of fairness as they
have been proposed in the multi-agent literature.  The problem of fair
allocation is being resolved in various MAS by using different
procedures, depending upon the technique of allocation of goods and
the nature of goods.  Its welfare implications in different systems
were explored by Rabin~\cite{16}.  Brams and Taylor give an analysis
of procedures for allocating divisible and indivisible items and for
resolving disputes among the self-interested agents~\cite{3}.  One of
the procedures described by them is the ``divide and choose'' method
of allocation of divisible goods among two agents to ensure the fair
allocation of goods which also exhibits the property of
``envy-freeness,'' a property first introduced by Foley~\cite{10}.
Lucas' method of markers, and Knaster's method of sealed bids are
described for MAS comprising more than two players and for the
division of indivisible items.  The adjusted-winner (AW) procedure is
also defined by Brams~\cite{2} for envy-freeness and equitability in
two-agent systems.  Various other procedures like moving-knife
procedures for cake cutting are defined for the MAS comprising three
or more agents~\cite{2,1}.

The auction mechanism proposed by Biggart~\cite{24} provides an
economic sociology perspective.  There, fairness can mean different
things for bidders and auctioneer.  The auctioneer may consider a
process fair which in fact only gives him the maximum revenue, whereas
the bidders may consider a process fair which only gives the
auctioneer the least return on all items.  The most important
consideration overall is to sustain the community's faith in the
fairness of the process.  This does not mean that buyers and sellers
cannot press their advantage, but they are allowed to do so only
insofar as the community as a whole considers their actions
appropriate and acceptable.

A concept of verifiable fairness in Internet auctions has been
proposed by Liao and Hwang~\cite{25}.  This was to promote trust in
Internet auctions.  The scheme proposed provides evidence regarding
policies implemented so that the confidence of bidders increases and
they consider it to be fair.  Most of these auctions see transparency
in the auctioning process and rules as the basis for ensuring fairness
in the system, but clarity regarding fairness still remains wanting.

The Nash bargaining concept is used by many economists.  In Nash
bargaining, there is no particular winner against a bargain.  If the
amount requested is within the total amount available at the owner
then they get their share, but if the demand is more then they get
nothing.  In our case, this is not the case as we have a winner in all
circumstances even if the auctioneer is facing a loss.  Also the
extended fairness concept is not present in Nash bargaining to acquire
a desired product whereas in our case bidders pay a price to achieve
extended fairness.

The game-theoretic concept of Shapley value~\cite{27} describes the
fairest allocation of collectively-gained profits between several
collaborative agents.  This is one approach used in coalitional games.
Though this deals with fair allocation, it is restricted to a
mechanism where the actors contribute in a coalition.  The profits
obtained are allocated in a fair manner.  The Shapley value is
different from our approach, as we do not take into account prior
understanding or coalition among the bidding agents in our discussion.

Fairness as a collective measure has been considered by
Moulin~\cite{26}, who proposes aggregate or collective welfare which
is measured in terms of an objective standard or index that assumes
equivalence between this measure and a particular mix of economic and
non-economic goods which gives happiness to a varying set of
individual utility functions.  This tries to capture social welfare
and commonwealth to be incorporated into every individuals' happiness
equations.  Though debatable, it provides an excellent introduction to
the concept of fairness.

A Distributed Combinatorial Auctioning System (DCAS) consisting of
auctioneers and bidders who communicate by message passing has been
proposed~\cite{30}.  Their work uses a fair division algorithm that is
based on DCAS concept and model.  It also discusses how basic and
extended fairness implementations may be achieved in distributed
resource allocation.

The fair package assignment model proposed by Lahaie and
Parkes~\cite{36} is defined on items having pure complements or super
additive valuations.  This model does not address combinatorial
package assignments which involve both complements and substitutes in
general.  Their model provides fairness to a ``core'' which contains a
set of all distributions which are considered competitive---no
fairness is posited for other distributions.  Hence the bidders whose
distributions lie outside the core do not get the benefits of fair
assessment.  In the case of multiple-round combinatorial auctions, for
example, bidders whose bids are not in the core during earlier rounds
are not in contention in later ones.  This scheme seems unfair in a
fundamental way, as it effectively discriminates against bidders who
cannot make it into the core.  In our model only truthfulness in
bidding is considered, and no bidders are distinguished based on
whether their bids lie inside or outside a putative core.

However, the term \emph{fairness} is defined differently in various
MAS with regard to the resource allocation.  In some MAS, it can be
defined as equitable distribution of resources such that each
recipient believes that he receives his fair share.  Thus, no agent
wants somebody else's share more than its own share; such division is
therefore also known as envy-free division of resources~\cite{2}.
Thus fair allocation is achieved if it is efficient, envy-free and
equitable~\cite{3}.

\section{Definitions of Fairness} \label{definitions}

Our additional notions of fairness in various MAS are basic fairness
and extended fairness.  This section defines the various notions about
fairness in combinatorial auctions in a MAS.

\subsection{Terminology} \label{term}

Let our CAS be a MAS which is defined by the following entities:

\begin{enumerate}

\item The total number of resources is represented by $m$ and the
  total number of bidders by $n$.

\item The set \(R= \{r_0, r_1, r_2, \ldots, r_{m-1}\}\) is a set of
  $m$ resources $r_i$, and \(2^R\) denotes the power set of $R$.

\item The set \(B= \{b_0, b_1, b_2, \ldots, b_{n-1}\}\) is a set of
  $n$ bidders $b_j$.

\item $a$ is the auctioneer who initially owns all the resources.

\item A package $S$ is some subset of the set of resources, i.e., \(S
  \subseteq 2^{R}\).

\item \(\mathbb{R}\) is the set of real numbers.

\end{enumerate}

For instance, consider a CAS that comprises three bidders $b_0$,
$b_1$, $b_2$, an auctioneer denoted as $a$, and three resources $r_0$,
$r_1$ ,$r_2$.  Each bidder is privileged to bid upon any combination
of these resources.  We denote the combinations or subsets of these
resources as \{$r_0$ \}, \{$r_1$\}, \{$r_2$ \}, \(\{r_0 , r_1\}\),
\(\{r_0 , r_2\}\), \(\{r_1 , r_2\}\), \(\{r_0 , r_1 , r_2\}\).

\begin{example} \label{ex3} 

  A package for a bidder winning the subsets \{$r_0$\} and \{$r_1$\}
  is defined as \(\{\{r_0\}, \{r_1\}\}\).

\end{example}

We also consider the concept of weight while assigning the fair value.
Here, weight is not the physical weight but is used as a
multiplicative factor for describing the desirability of the package
by the bidder.  If a higher weight is assigned to a package, then it
will result in a higher fair value.  This expresses the well-known
fact that a bidder is likely to assign a higher fair value to a
resource that is desired or needed than to one that is not, even when
the two resources have the same intrinsic value (e.g., a starving man
is likely to assign a far higher value to a meal than to any other
commodity of equivalent intrinsic worth).

\begin{definition} \label{def3}

  Let us define some important terms used in our later discussion, as
  follows:

\begin{enumerate}

\item The \emph{initial value} of an item is defined as \(\Omega: B \times R
  \rightarrow \mathbb{R}\), where \(\Omega(b_i,r_j)\) is the initial
  amount attached by bidder $b_i \in B$ to a resource $r_j \in R$.

\item The \emph{weight}  a package is defined as \(\Theta: B \times 2^{R}
  \rightarrow \mathbb{R}\), where \(\Theta(b_i,S)\) is the weight for
  bidder $b_i \in B$ of package $S$.

\item The \emph{fair value} for a resource is defined as \(\Pi: B
  \times R \rightarrow \mathbb{R}\), where \(\Pi(b_i,r_j) =
  \Theta(b_i,{r_j}) \times \Omega(b_i,r_j)\) for a bidder $b_i \in B$
  on resource $r_j \in R$.
   
  The fair value for a package is defined as \(\Pi: B \times 2^{R}
  \rightarrow \mathbb{R}\), where \(\Pi(b_i,S)\) is the value obtained
  as \(\sum \Pi(b_i,r_j), \forall r_j \in S \).

\item The \emph{bid value} of a package is defined as \(\Upsilon: B
  \times 2^{R} \rightarrow \mathbb{R}\), where \(\Upsilon(b_i,S)\) is
  the amount that the bidder $b_i \in B$ is willing to give in
  exchange for the package $S$.

\item The \emph{utility value} of a package is defined as \(\Gamma: B
  \times 2^{R} \rightarrow \mathbb{R}\), where \(\Gamma(b_i,S) =
  \Upsilon(b_i,S) - \Pi(b_i,S)\). 

\item The \emph{package cost} is defined as \(\Psi: B \times 2^{R}
  \rightarrow \mathbb{R}\), where \(\Psi(b_i,S)\) gives the final
  winning amount for bidder $b_i$ on package $S$ after the bidding has
  ended.

\end{enumerate}  

\end{definition}

\eat{

(The above definitions are not connected into a single equation as the 
are independent entities which do not interact together at the same 
time as discussed in later sections.)\\

}

 Assume that the auctioneer and each bidder all have fair values for
each of the individual resources (say, in dollars) as shown in
Table~\ref{tab:fairtable}.  Every bidding process will have a base
value initially assigned to an item from where the bidding proceeds.
The fair values by a bidder and an auctioneer for each resource
represent their measures of its actual value, and depend on their
weights and their initial values (Definition~\ref{def3}).  Thus, a bidder
is willing to consider a resource at his fair value.  Similarly, the
auctioneer is willing to sell a resource at \textsl{his} fair value.
However, bid value may be higher or lower than fair value and hence result 
in higher or lower utility values (Definition~\ref{def3}) depending on the 
need of the resource. Fair value for a combination of resources in the fairness 
table can be calculated as the sum of the fair value for each of the resources in
that combination (fair values are considered additive where as the bid values 
are combinatorial in nature and not additive).

\begin{table}[ht]
\centering
\begin{tabular}{c c c c} \toprule
\multicolumn{4}{c}{Example Fairness Table} \\ \cmidrule(r){1-4}
& Resource $r_0$ & Resource $r_1$ & Resource $r_2$\\ \midrule
Bidder $b_0$ &5  &8 &8 \\ 
Bidder $b_1$ &10  &2 &8\\
Bidder $b_2$ &10  &5 &10\\
\addlinespace
Auctioneer $a$ &8  &10 &15\\ \bottomrule
\end{tabular}
\caption{Fair valuations for each resource by all bidders and auctioneer}
\label{tab:fairtable}
\end{table}

From Table~\ref{tab:fairtable}, we can see that bidder $b_0$ values
resource $r_0$ at \$5, $r_1$ at \$8 and $r_2$ at \$8.  This means that
bidder $b_0$ is willing to pay \$5 for $r_0$, \$8 for $r_1$, and \$8
also for $r_2$; $b_0$ believes that no loss is incurred by the
auctioneer in this trade.  The fair value for the subset \(\{r_0 , r_2
\}\) for the bidder $b_0$ is calculated as the sum of his the fair
values for $r_0$ and $r_2$, i.e., \$5 + \$8 = \$13.  Similarly, the
fair value for a package is the sum of the fair values of the
comprising sets (Definition~\ref{def3}), i.e., for a package
\{$r_0$\}, \(\{r_1 , r_2\}\), the fair value is the sum of the fair
values of \{$r_0$\} and \(\{r_1 , r_2\}\).

A bidder participates in the bidding process by quoting his bid for
the packages.  Let the bids raised by the bidders for the individual
resource and different combination of resources be as given in Table
$2$. It can be seen that the bids raised by each of the bidder for
different sets of resources may or may not be equal to the fair value
of the respective set of resources. This is because the combinations
may be complimentary or substitutes.

A bidder is considered to make bid zero for any sets of resources he
does not wish to procure.

\begin{table}[ht]
\centering
\begin{tabular}{c c c c c c c c} \toprule 
\multicolumn{8}{c}{Example Bid Table} \\ \cmidrule(r){1-8}
& $r_0$ &$r_1$ &$r_2$ &\{$r_0$ , $r_1$\} &\{$r_0$ , $r_2$\}  &\{$r_1$ , $r_2$\} &\{$r_0$ , $r_1$ , $r_2$\}\\ \midrule
Bidder $b_0$ &0 &10 &5 &10 &20 &15 &50\\ 
Bidder $b_1$ &10 &5 &10 &30 &0 &0 &50\\
Bidder $b_2$ &10 &0 &15 &20 &30 &15 &30\\	\bottomrule	
\end{tabular}
\caption{Bids raised by the bidders for different combination of resources}
\label{tab:bidtable}
\end{table}

With this terminology we proceed to explain fairness in subsequent
sections.

\subsection{Basic Fairness} \label{Basic}

In many MAS, there occurs a need of allocating the resources in an
equitable manner, i.e., each agent gets an equitable share of the
resources.  Such allocations leave the agents with a feeling that they
have received a fair share~\cite{3}.  For example, if we consider a
method that would leave two agents feeling as if they had received
$60$\% of the good then we would call it equitable.  If one felt to be
favored and had received $80$\% while the other agent believed to have
received $60$\% then it would not be equitable~\cite{3}.  This is
quite difficult to access and tends to quite subjective in many cases.
We give a mechanism where this applies only in case of a tie, hence we
consider a divisible resource which does not lose its value upon
division and divide it equitably among bidders in proportion to their
assigned weights.  Each agent has a set of allocations he deems fair.
An allocation is then is said to achieve basic fairness in resource
allocation if all agents deem it fair.

Each bidder $b_i$ wants to maximize his chances of procuring the resource and individual 
utility given by \(\Gamma(b_i,S)\) represents the satisfaction of obtaining the resource.  
The most simple approach is that the satisfaction
of a bidder does not depend on other bidders' satisfactions.  The
representation below considers that a package $S$ is divisible and can
be divided equitably among $n$ bidders in proportion to their utility values.

The resource can be divided equitably in the ratio: \(\Gamma(b_i,S) / \sum_{i=1}^{n}\{\Gamma(b_i,S)\}\), 
where weights are set freely by agents.

\begin{definition} \label{def1}

  If each bidder $b_i$ has a utility for a package $S$ given
  by \(\Gamma(b_i,S)\), and the package $S$ can be divided equitably
  among $n$ bidders in the ratio \(\Gamma(b_i,S) / \sum_{i=1}^{n}\{\Gamma(b_i,S)\}\), 
  then basic fairness is said to be achieved.

\end{definition}

\begin{example} \label{ex1}

  Consider there to be three bidders for the divisible package
  $S$.  The bidders' bid values, fair values and utility values 
  are shown in Table~\ref{tab:basicfair}.

\begin{table}[ht]
\centering
\begin{tabular}{c c c c} \toprule
\multicolumn{4}{c}{Utility and Weight Table} \\ \cmidrule(r){1-4}
& Bidder $b_0$ & Bidder $b_1$ & Bidder $b_2$\\ \midrule
Bid Value &$24$  &$16$ &$20$ \\ 
Fair Value &$18$  &$12$ &$16$ \\ 
Utility Value &$6$  &$4$ &$4$ \\\bottomrule
\end{tabular}
\caption{Example to demonstrate basic fairness}
\label{tab:basicfair}
\end{table}

The calculations of ratios are done as shown below.

For bidder $b_0$, $6$/$14$ = $0.43$.

For bidder $b_1$, $4$/$14$ = $0.285$.

For bidder $b_2$, $4$/$14$ = $0.285$.

If the winning amount is \$$100$ then it is divided in the ratio
$0.43 : 0.285 : 0.285$ to achieve basic fairness, i.e., bidder $b_0$ has
to pay \$$43$, bidder $b_1$ has to pay \$$28.50$, bidder $b_2$ has
to pay \$$28.50$.
\end{example}

This method of equitable allocation ensures that all agents deem the allocation
to be fair.  Therefore, we say that every agent believes that the set
of resources is divided fairly among all the agents.  This concept of
fairness is termed as basic fairness.

This kind of fairness is required in applications wherein fairness is
the key issue, rather than the individual satisfactions of the
self-interested agents.  In such applications, it becomes necessary to
divide a package in an equitable fashion so that every agent believes
that it is receiving its fair share from the set of resources. Hence,
we see that every agent enjoys material equality and this ensures
basic fairness among them.


\subsection{Extended Fairness} \label{extended}

In order to ensure egalitarian social welfare~\shortcite{8}, basic
fairness is alone not sufficient. We also need to address
envy-freeness~\cite{2}.  Envy-free allocations result in each agent
being at least as happy with its share of the goods as it would be
with any of the other agents shares despite the difference in some
payments for identical goods~\cite{3,29}.  Here, we need to ensure
that the allocation is perceived by all agents to be a fair
allocation.

In a MAS, every agent assigns a fair value to each resource that
determines its estimate of the value of the resource in quantitative
terms.  The fair value attached to each resource can be expressed in
monetary terms in most MAS.  Here, the agent believes that the
allocated resource is fair if he receives the entire allocation and
the value is according to his fair estimate.  

However, it is important
to mention that the fair value attached to each resource by an agent
does not necessarily reflect the bid value of the resource.  An
agent may hold a higher or lower bid value for a resource
irrespective of the fair value attached to the resource.  Rather, the
fair value attached to a resource is an estimate of the actual value
of the resource in the system as perceived by an agent in quantitative
terms.  It means that an agent is always willing to trade a resource
at its fair value.

Let there be $k$ bidders who bid for a package $S$.  Let each bidder
$b_i$ and auctioneer $a$ (who is here only as a facilitator for
achieving the items' fair value) give their fair values in the
fairness table (as in the example in Table~\ref{tab:fairtable}), which
is open for all to see at the \textsl{end} of the bidding process.

\begin{definition} \label{def4}

  Let us define some of the terms used in our discussion.

\begin{enumerate}

\item $C$ is defined as the winning amount after the bidding process
  for a package $S$.
\item \(\xi: {a} \times 2^{R} \rightarrow \mathbb{R}\) defines the
  fair value of the auctioneer $a$ for a package $S$ denoted by
  \(\xi(S)\).
\item \(\Pi(b_i,S)\) is defined as the fair value of the bidder $b_i$
  for a package $S$.
\item The \emph{profit} denoted by $\Phi$ is defined as the net amount
  $C$ above the fair value of the auctioneer \(\xi(S)\) given by the
  bidder $b_i$ for the package $S$ and is calculated as difference
  $C-\xi(S)$.
\item The function \(distribute\) is defined as the amount \(x: B \times 2^{R}
  \rightarrow \mathbb{R}\) to be given back to the losing bidders $b_i$ who bid for the winning package $S$.
\item The value $reward$ is defined as \(reward: B \times 2^{R}
  \rightarrow \mathbb{R}\), where \(reward = \Phi - x\).

\end{enumerate}

\end{definition}

Now let us define extended fairness in resource allocation.

\begin{definition} \label{def2}
  An allocation is said to satisfy extended fairness, if when a winning bidder $b_i$ is allocated a
  package $S$: (i) if \(\Upsilon(b_i,S) > \xi(S)\), then a losing bidder $b_j$ is rewarded \(\Phi \times \left(\frac{\Pi(b_j,S)-\xi(S)}{\xi(S)}\right)\); and (ii) if \(\Upsilon(b_i,S) \leq \xi(S)\), then no one gets a reward.

\end{definition}

Consider the following scenarios:
\begin{enumerate}

\item The auctioneer makes a profit more than his fair value assigned
  initially for that package $S$.  He distributes the profit among the
  losing bidders in proportion to their fair values for that package
  $S$ as follows:

  Let $C$ be the winning bid which is greater than the fair value of
  the auctioneer, i.e., $\xi(S)$.  Therefore, \(\Phi = C-\xi(S)\) by
  Definition \ref{def4}.  The profit to be distributed for each losing
  bidder $i$ is calculated by:

\[distribute \left(\Phi \times \left(\frac{\Pi(b_i,S)-\xi(S)}{\xi(S)}\right)\right)\]

Now, the incentive for the winning bidder is \(reward = \Phi -
distribute\left(\Phi \times \left(\frac{\Pi(b_i,S)-\xi(S)}{\xi(S)}\right)\right)\).  Thus, the
auctioneer has obtained his fair value and hence considers this
allocation as fair.  All the bidders get amounts according to their
fair values, which makes them envy free.

\item The auctioneer gets a winning bid $C$ which is exactly the same
  as the fair value $\xi(S)$ associated with the package $S$.  Now the
  profit is zero. Therefore, the auctioneer has obtained his fair
  value and hence considers this allocation as fair.  All the bidders,
  though did not get any reward consider this allocation as envy-free
  as auctioneer too did not make any profits more than his own fair
  value.

\item The auctioneer gets a winning bid $C$ less than the fair value
  $\xi(S)$ attached by him for the package $S$.  In this case, we try
  to minimize his loss as follows:

  If the fair value given by bidder $\Pi(b_i,S) \geq \xi(S)$, then
  bidder $b_i$ pays $\xi(S)$.  Thus, the auctioneer has no loss as he
  gets his fair value and the bidder too is envy free since he
  considers that paying his fair value as fair.  The other bidders are
  still envy free since the amount paid by the winning bidder is more
  than he actually won in the bidding process.

  If the fair value given by bidder $\Pi(b_i,S) < \xi(S)$ and
  $\Pi(b_i,S) \leq C$, then the payment does not change and he pays
  $C$, else he pays $\Pi(b_i,S)$.  Only in this case auctioneer fails
  to get his fair value and the bidder does not get the distributed
  profit amount.  The allocation is still envy-free for all the
  bidders but not for the auctioneer.  This can be avoided if both
  bidder and auctioneer remain truthful in their fair values.

\end{enumerate}

\begin{example} \label{ex2}

  An example of such a system can be explained with a scenario of
  auctioning of a painting. The contending bidders express their fair
  values through their sealed bids that is submitted to the
  auctioneer, i.e., each contending bidder believes that his quotation
  fulfills the value expected by the auctioneer and he is a
  competitive contender for the painting. We assume here that all
  bids are truthful. An unbiased auctioneer selects the bid which is
  the maximum for revenue maximization and the painting is allotted to
  him. Here the auctioneer would distribute the profits among losing
  bidders when he gets back his fair value. This takes care of
  envy-freeness. Hence, the allocation is perceived to be fair by the
  winning bidder and by all other bidders as it is allocated to the
  most deserving among all the bidders. The auctioneer also perceives 
  this to be fair since he will obtain the fair value for the
  resource. Thus all participants perceive the allocation to an agent
  to be fair irrespective of the fair values attached by
  them. Therefore, extended fairness is said to be achieved.

\end{example}

To make these notions of fairness mathematically precise, we need a
framework where fairness is a required property in resource
allocation. However, we also see that resource allocation deals with
another key issue of optimality in various MAS. The best example of
resource allocation framework where both optimality and fairness are
the key issues is Combinatorial Auctioning Systems (CAS).

\section{Fairness in Combinatorial Auctioning Systems
  (CAS)} \label{fairCAS}

Combinatorial Auctioning Systems are a kind of MAS which comprise an
auctioneer and a number of self-interested bidders. The auctioneer
aims at allocating the available resources among the bidders who, in
turn, bid for sets of resources to procure them in order to satisfy
their needs. The bidders aim at procuring the resources at minimum
value during the bidding process, while the auctioneer aims at
maximizing the revenue generated by the allocation of these
resources. Thus, CAS refers to a scenario where the bidders bid for
the set of resources and the auctioneer allocates the same to the
highest-bidding agent in order to maximize the revenue. Hence, we see
that optimality is one of the key issues in CAS.

An algorithm of Sandholm is used here to attain optimal allocation of
resources.  Sandholm proposes various methods for winner determination
in combinatorial auctions~\cite{18}. The search methodology can be
used to obtain optimal allocation of resources. We can represent the
Table \ref{tab:bidtable} as a Bid tree using an algorithm of
Sandholm~\cite{18}. We can also carry out some preprocessing steps to
make the steps faster without compromising the
optimality~\cite{13,18}. Thus we can determine the winning bidders.

However, besides optimality, another key issue desired by some
auctioning systems is fairness. To incorporate this significant
property in this resource allocation procedure, we propose an
algorithm which uses the concept of extended fairness for each agent
with basic fairness in case of a tie and determines the final payment
made by the winning bidders.

The algorithm that we describe is based upon a CAS that uses an
algorithm of Sandholm for achieving optimality, and an incentive
compatible mechanism called Generalized Vickrey Auction (GVA) along with Algorithm \ref{Algorithm1} as the
pricing mechanism that determines the payments to be given by the
winning bidders. 

The Generalized Vickrey Auction (GVA) has a payoff
structure that is designed in a manner such that each winning agent
gets a discount on its actual bid. This discount is called a Vickrey
Discount, and is defined by~\cite{13} as the extent by which the total
revenue to the seller is increased due to the presence of that winning
bidder, i.e., the marginal contribution of the winning bidder to the
total revenue.  The GVA framework requires significant transfer payments 
from bidders to auctioneer hence a redistribution mechanism is required 
to reduce the cost of implementation~\cite{37,38}. Hence, after we obtain 
winning bidders from the algorithm of Sandholm, the GVA mechanism can be 
applied to get Package Cost (Definition \ref{def3}) and Algorithm \ref{Algorithm1} 
can be used for redistribution of payments back to the bidders to achieve fair 
allocation. We give mathematical formulations to show that both kinds of fairness 
can be achieved in CAS.

\subsection{Notion of Fairness in Combinatorial Auctions}

\begin{enumerate}

\item Each bidder and the auctioneer define its fair values in the
  fairness table (Table \ref{tab:fairtable}) before the start of the
  bidding process.  It is a sealed matrix and is unsealed at the end
  of bidding process.
\item An allocation tree is constructed at the end of the bidding
  process to determine the optimum allocation and the winning
  bidders~\cite{18}.  Information about all the bidders in a tie is
  not discarded using some pre-defined criteria.
\item Calculate the package cost \(\Psi(b_i,S_j)\) (Definition \ref{def3}) 
denoted by $P_{i,j}$ which is the final winning bid amount for bidder $b_i$ on 
package $S_j$ is obtained after applying GVA scheme.
\item The fair value of the package won by each bidder is calculated,
  and the value is denoted as $Q_{i,j}$ for the bidder $b_i$ who wins
  the package $S_j$.
\item The fair value of each package is calculated using the fairness
  table of the auctioneer and is denoted as $Q_{a,j}$ for a package
  $S_j$.
\item The values of $Q_{a,j}$ and $P_{i,j}$ are compared to determine
  the final payment by the bidder which is considered fair.
\end{enumerate}

Now, we propose an algorithm which satisfies extended fairness in all
cases, except in case of a tie.

\subsection{Notations Used in
  Algorithm~\ref{Algorithm1}} \label{notations}

\begin{itemize}
\item $b_i$ is an arbitrary bidder $i$ who belongs to the set of
  bidders $B$.
\item $S_j$ is the winning package with complimentaries and substitutes
  included, which is a subset of set $R$.
\item In general, the Fairness Table for bidder $b_i$ and Auctioneer
  $a$ is defined as shown in Table~\ref{tab:generalfairtable}.

\begin{table}[ht]
\centering
\begin{tabular}{c c c c c c} \toprule
\multicolumn{6}{c}{Fairness Table} \\ \cmidrule(r){1-6}
& Resource $r_0$ & Resource $r_1$ & Resource $r_2$ & \ldots & Resource $r_{m-1}$\\ \midrule
Bidder $b_0$ & \(\Pi(b_0,r_0)\)  & \(\Pi(b_0,r_1)\)  & \(\Pi(b_0,r_2)\)  & \ldots  & \(\Pi(b_0,r_{m-1})\) \\ 
Bidder $b_1$ & \(\Pi(b_1,r_0)\)  & \(\Pi(b_1,r_1)\)  & \(\Pi(b_1,r_2)\)  & \ldots  & \(\Pi(b_1,r_{m-1})\) \\
. & . & . & . & \ldots  & .  \\. & . & . & . & \ldots  & .  \\. & . & . & . & \ldots  & . \\
Bidder $b_{n-1}$ & \(\Pi(b_{n-1},r_0)\)  & \(\Pi(b_{n-1},r_1)\)  & \(\Pi(b_{n-1},r_2)\)  & \ldots  & \(\Pi(b_{n-1},r_{m-1})\) \\
\addlinespace
Auctioneer $a$ & \(\xi(r_0)\)  & \(\xi(r_1)\)  & \(\xi(r_2)\)  & \ldots  & \(\xi(r_{m-1})\) \\ \bottomrule
\end{tabular}
\caption{Fair values for each resource by all bidders and auctioneer}
\label{tab:generalfairtable}
\end{table}

\item The fair value function by a bidder $b_i$ for a resource $r_j$
  is given by \(\Pi (b_i, r_j) = d\), where d $\in \mathbb{N}$.
\item $Q_{i,j}$ is the fair value of resource $r_j$ by bidder $b_i$
  where $r_j\in$ R and $b_i \in$ B.
\item $Q_{a,j}$ is the fair value of resource $r_j$ by auctioneer $a$
  where $r_j\in$ R. Here we consider only a single auctioneer.
\item The package Cost $P_{i,j}$ (Definition \ref{def3}) for bidder
  $b_i$ obtained from the GVA scheme is represented as \(\Psi
  (b_i,S_j)\).(The package cost is a function of bid values on the bundles of resources)
\item The pay function by a bidder $b_i$ is represented as $pay(c)$ is
  the final payment to be made to the auctioneer by the bidder $b_i$
  where $c$ is the bid amount.
\item $\Phi$ (Definition~\ref{def4}) is the net amount above the fair
  value distributed by the auctioneer $a$ to the bidders for a package
  $S_j$.
\item \(distribute\) is a function which calculates the amount to be
  given back to the bidders who bid for the winning package $S_j$
  (Definition~\ref{def4}).
\item $loss$ is the net amount below the fair value given by the
  auctioneer to the package $S_j$.
\end{itemize}

\subsection{Flow of Algorithm~\ref{Algorithm1}} \label{flow}

In Algorithm~\ref{Algorithm1}, we calculate the package cost and the
fair values of bidder and auctioneer given in the lines 1--3.  These
are calculated in the beginning and are represented as $P_{i,j}$,
$Q_{i,j}$ and $Q_{a,j}$ respectively.

In lines 4--9, we have the first {\bf if} condition where the package
cost $P_{i,j}$ is greater than the fair value assigned by the
auctioneer for the package $Q_{a,j}$. If this evaluates to
\textsc{TRUE}, then the bidder pays the amount but the net profit
calculated is distributed among all the bidders who bid for that
package proportional to their bids.
\\\\
In lines 10--12, we have the second {\bf if} condition where the
package cost $P_{i,j}$ is equal to the fair value assigned by the
auctioneer for the package $Q_{a,j}$. If this evaluates to
\textsc{TRUE}, since there is no profit the bidder still pays and
there is no amount distributed to the winning package bidders.
\\\\
In lines 13--32, we have the third `if' condition where the package
cost $P_{i,j}$ is less than the fair value assigned by the auctioneer
for the package $Q_{a,j}$. If this evaluates to \textsc{TRUE}, there
is a loss for the auctioneer so we try to minimize the loss by
checking the additional cases as follows.
\\\\
First at line $15$, if the fair value of bidder $Q_{i,j}$ is greater
than fair value of auctioneer $Q_{a,j}$ evaluates to \textsc{TRUE},
then the bidder will have to pay only $Q_{a,j}$. This prevents loss
for auctioneer and also the bidder deems it as fair.
\\\\
Secondly at line $19$, if the fair value of bidder $Q_{i,j}$ is equal
to the fair value of auctioneer $Q_{a,j}$ evaluates to \textsc{TRUE},
then the bidder will have to pay only $Q_{a,j}$ as in the previous
condition. Similar to the previous condition this prevents loss for
auctioneer and also the bidder deems it as fair.
\\\\
Finally at line $23$, if the fair value of bidder $Q_{i,j}$ is less
than fair value of auctioneer $Q_{a,j}$ evaluates to \textsc{TRUE},
then we have to see the additional two conditions as follows.
\\\\
If the fair value of bidder $Q_{i,j}$ is less than or equal to package
cost of bidder $P_{i,j}$ then the bidders' final payment remains the
same, i.e., $P_{i,j}$.
\\\\
If the fair value of bidder $Q_{i,j}$ is greater than the package cost
of bidder $P_{i,j}$ then the bidders' final payment is $Q_{i,j}$.
These are presented in Algorithm~\ref{Algorithm1}.

\subsection{Algorithm to Incorporate Extended Fairness} \label{algo}
\begin{algorithm}[H]
\SetLine
\linesnumbered
\KwData{package cost, fair value of winning bidder, fair value of auctioneer}
\KwResult{Final Payment by the bidder }
{\(P_{i,j} \leftarrow \Psi(b_i,S_j)\)}\tcc*[f]{where \(P_{i,j} \in \mathbb{R}\)}\\
{\(Q_{i,j} \leftarrow \Pi(b_i,S_j)\)}\tcc*[f]{where \(Q_{i,j} \in \mathbb{R}\)}\\
{\(Q_{a,j} \leftarrow \xi(S_j)\)}\tcc*[f]{where \(Q_{a,j} \in \mathbb{R}\)}
		
\If{\(P_{i,j} > Q_{a,j}\) }{
\(pay(P_{i,j})\)\;
\(\Phi \leftarrow (P_{i,j} - Q_{a,j})\)\;
\(distribute(\Phi\times[(Q_{k,j} - Q_{a,j}) / Q_{a,j}] )\)\
\tcc*[r]{among other bidders who bid for package $S_j$}\
}
\If{\(P_{i,j} = Q_{a,j}\) }{
\(pay(P_{i,j})\)\;
}
\If{\(P_{i,j} < Q_{a,j}\) }{
\(loss \leftarrow (Q_{a,j} - P_{i,j})\)\tcc*[r]{Auctioneer can recover as follows}\
\If{\(Q_{i,j} > Q_{a,j}\)}{
\(pay(Q_{a,j})\)\tcc*[r]{Bidder's estimate of fair value is more than $P_{i,j}$}\
}
\If{\(Q_{i,j} = Q_{a,j}\) }{
\(pay(Q_{a,j})\)\tcc*[r]{fair value is same and $Q_{i,j}$ greater than $P_{i,j}$}\
}
\If{\(Q_{i,j} < Q_{a,j}\) }{
	
\eIf{\(Q_{i,j} <= P_{i,j}\) }{
\(pay(P_{i,j})\)\tcc*[r]{Bidder's final payment remains the same}\
}{
\(pay(Q_{i,j})\)\tcc*[r]{Bidder's final payment is $Q_{i,j}$}\
}
}
}	
\caption{Algorithm incorporating extended fairness}
\label{Algorithm1}
\end{algorithm}
		
\subsection{Handling a Case of a Tie---Incorporating Basic Fairness}

Unlike traditional algorithms, we do not discard the bids in the case
of a tie on the basis of some pre-decided criterion.  We consider
these cases in our algorithm to provide basic fairness to the bidders.
In case of a tie, we shall measure the utility value of the resource
to each bidder in the tie.  The utility value of a resource to a
bidder is the quantified measure of satisfaction or happiness derived
by the procurement of the resource (Definition~\ref{def3}).

The bidders maximize this utility value to quantify the importance and
their need for the resource.  Thus, the higher the utility value, the
greater is the need for the package.  In such a case, fairness can be
achieved if the package $S$ is divided among all the bidders in a
proportional manner, i.e., in accordance to the utility value attached
to the package by each bidder.

\begin{example} \label{ex5}

  Let us consider the same example to explain the concept of basic
  fairness in our system.  From Table~\ref{tab:bidtable}, we observe
  that the optimum allocation attained through allocation tree
  comprises the package {$r_0$, $r_1$ , $r_2$ } as it generates the
  maximum revenue of \$$50$.  However, we see that this bid is
  submitted by the two bidders, $b_0$ and $b_1$.

  Thus, we calculate the fair value of the package $S_j$ = \{$r_0$,
  $r_1$ , $r_2$\} for the bidder $b_0$ and $b_1$, i.e., \(\Pi(b_0,
  S_j) = 5+8+8 = \$21\) and \(\Pi(b_1, S_j) = 10+2+8 = \$20\).  Thus,
  the utility value of the package $S$ for the bidder $b_0$ and $b_1$
  is as follows:\textsc{}

  For bidder $b_0$ , \(\Gamma(b_0, S_j) = 50 - 21 = \$29\), and

  For bidder $b_1$ , \(\Gamma(b_1, S_j) = 50 - 20 = \$30\).

  Hence, the package $S$ is divided among bidders $b_0$ and $b_1$, in
  the ratio of $29:30$.  In other words, bidder $b_0$ gets $49.15$\%
  and bidder $b_1$ gets $50.85$\% of the package $S_j$.

  The payment made by the bidders is also done in the similar
  proportional manner similar to Example~\ref{ex1}.

  The bidders $b_0$ and $b_1$ make their respective payments in the
  ratio of $29:30$ to make up a total of \$$50$ for the auctioneer,
  i.e., bidder $b_0$ pays \$$24.65$ and bidder $b_1$ pays \$$25.35$ to
  the auctioneer for their respective shares.

\end{example}

Hence, we see that extended fairness as well as basic fairness are
achieved in a CAS using our approach.  We take into account the fair
estimates of the auctioneer and the bidders for each resource to
ensure that fairness is achieved to auctioneer as well as the bidders.
A detailed analysis of our mechanism is in the following section.

\section{Analysis} \label{analysis}

Using the solution concept of dominant strategies and mechanism design
with quasilinear preferences, we can analyze the following.

We say that the agents' preferences are quasilinear when they satisfy
the conditions given below: first we are in a setting where the
mechanism can choose to charge or reward an agent an arbitrary amount.
Second, and more restrictive, is that an agent's utility of a choice
cannot depend on the money that he has, i.e., his value is the same
whether he is rich or poor.  Finally, the agents care only about the
choice selected and their own payments, i.e., they are not concerned
about monetary payments made or received by other agents.

\subsection{Fairness} \label{fairness} 

We say that extended fairness is achieved when a bidder procures a
resource for an amount that is equal to his estimate of fair value of
that resource. In such a case, the bidder believes that the resource
was procured by it at a fair amount irrespective of other bidders
estimate of fair value of that resource. This is according to the last
condition of quasilinear preference. Thus, the allocation is believed
to be extendedly fair as per the estimates of the winning bidder.

We also see that basic fairness is achieved in our system when there
is more than one bidder who has raised equal bid for the same set of
resources.  In such a case, we divide the set of resources among all
the bidders so as to ensure fairness to all the bidders in a
tie. However, this division of resources set is done in a proportional
manner. We intend to divide the resource such that the bidder holding
highest utility value to it should get the biggest share.  To ensure
this, we calculate the utility value (i.e., \(\Gamma(b_i, S_j) =
\Upsilon(b_i,S_j) - \Pi(b_i,S_j))\) of the set of resources to each
bidder and divide the set in the ratio of these values among the
respective bidders. Thus, we see that each bidder procures his basic
share of the set of resources in accordance to the basic importance
attached by the bidder to the set of resources. Due to the achievement
of fairness through our payment scheme, the bidders are expected to
show willingness to participate in the auctions.

\subsection{Higher Rewards} \label{rewards} 

Here we show that a bidder is encouraged to bid higher as he gets
rewards proportional to his bids which are fair.

\begin{theorem} \label{Thm1}

  Given a CAS, a bidder has an incentive to bid higher using the
  extended fairness algorithm \ref{Algorithm1} as he gets higher
  rewards which are fair provided:

\begin{enumerate}
\item the bid value $P_{i,j}$ is always greater than or equal to fair
  value $Q_{a,j}$ on package $S_j$.
\item $Q_{i,j}$ of winning bidder is always greater than or equal to
  $Q_{i,j}$ of the losing bidder.
\end{enumerate}

\end{theorem}

\begin{proof}

  Assume any two bidders \(b_x, b_y \in B\) who bid for package
  $S_j$.  
  Assume values: \(P_{x,j},Q_{x,j}\) for bidder $b_x$.
  \(P_{y,j},Q_{y,j}\) for bidder $b_y$.  \(Q_{a,j}\) for auctioneer
  $a$.

  The first condition is \(P_{x,j} > Q_{a,j}\) and \(Q_{x,j} >
  Q_{y,j}\) hence we have a profit $\Phi_x$ which is distributed in
  the ratio \(\Phi_x\times \left(\frac{Q_{x,j} -
      Q_{a,j}}{Q_{a,j}}\right)\) to $b_x$ and
  \(\Phi_x\times\left(\frac{Q_{y,j} - Q_{a,j}}{Q_{a,j}}\right)\) to
  $b_y$. This gives proportional as well as fair incentives to $b_x$
  and $b_y$. This also gives a higher reward to $b_x$ since \(Q_{x,j}
  > Q_{y,j}\).

  For notational convinience, let $k$ represent \(\frac{Q_{x,j} - Q_{a,j}}{Q_{a,j}}\) and $l$ represent \(\frac{Q_{y,j} -
    Q_{a,j}}{Q_{a,j}}\).

  The second condition is \(P_{x,j} > Q{a,j}\) and \(Q_{x,j} >
  Q_{y,j}\) hence we have a profit $\Phi_y$ which is distributed as
  $\Phi_y \times k$ and $\Phi_y \times l$ where $k, l$ are constant
  for $b_x$ and $b_y$. When \(\Phi_y > \Phi_x\) we see a greater
  amount of reward is given to higher bidding and hence bidders are
  encouraged to bid more. \qedhere

\end{proof}

\begin{theorem} \label{Thm2}

  Given a CAS, a bidder gets higher rewards using the extended
  fairness algorithm \ref{Algorithm1} if his fair value $Q_{i,j}$
  satisfies the condition \(Q_{a,j}> Q_{i,j} > 2\times Q_{a,j}\) for a
  package $S_j$.

\end{theorem}

\begin{proof}

  Assume a bidder \(b_i \in B\) who bids for package $S_j$ and wins
  it.  Assume values: \(P_{i,j},Q_{i,j}\) for bidder $b_i$.
  \(Q_{a,j}\) for auctioneer $a$.
      
  The first condition is when the winning bidder $b_i$ has given a a
  low fair value for a package $S_j$ intentionally, i.e., \(Q_{i,j} <
  Q_{a,j}\).  Now, his ratio is calculated as \(k=\frac{Q_{i,j} -
    Q_{a,j}}{Q_{a,j}}\) which is negative.  He has to distribute an
  extra amount of the same proportion, i.e.,
  \(distribute\left(\Phi\times\left(2 \times \frac{Q_{i,j} -
        Q_{a,j}}{Q_{a,j}}\right)\right)\).  Hence, \(reward=\Phi -
  distribute \left(\Phi\times \left(2 \times \frac{Q_{i,j} -
        Q_{a,j}}{Q_{a,j}}\right)\right)\). Therefore, bidder has to
  pay $P_{i,j} - reward$.
  
  The second condition is when the winning bidder $b_i$ has given a
  very high fair value for a package $S_j$, i.e., \(Q_{i,j} > 2 \times
  Q_{a,j}\). Now, his ratio is calculated as \(k=\frac{Q_{i,j} -
    Q_{a,j}}{Q_{a,j}}\) which is greater than 1. He has to distribute
  \(\Phi\times\left(\frac{Q_{i,j} - Q_{a,j}}{Q_{a,j}}\right)\) which
  is greater than $\Phi$. Therefore, the extra amount to be
  distributed would be added to $P_{i,j}$ and hence ends up paying
  higher amount without any rewards.
  
  The third condition is when the winning bidder $b_i$ has given a
  true fair value for a package $S_j$, i.e., \(Q_{a,j}>Q_{i,j} > 2
  \times Q_{a,j}\). Now, his ratio is calculated as \(k=\frac{Q_{i,j}
    - Q_{a,j}}{Q_{a,j}}\) which is a proper fraction. Hence,
  \(reward=\Phi - distribute \left(\Phi\times \left(\frac{Q_{i,j} -
        Q_{a,j}}{Q_{a,j}}\right)\right)\).  Therefore, bidder $b_i$
  has to pay $P_{i,j} - reward$ from definition \ref{def4}.
   
  Clearly, we can see that the maximum reward is possible only in the
  third condition, where the fair value is neither too low nor very
  high.  Thus, Algorithm~\ref{Algorithm1} provides higher rewards if
  fair value for a package $S_j$ satisfies the condition \(Q_{a,j}>
  Q_{i,j} > 2\times Q_{a,j}\). \qedhere
  
\end{proof}

\subsection{Other Issues} \label{otherissues}

Now let us discuss some issues considering the quasilinear mechanism.

\subsubsection{Truthfulness} \label{truth}

Consider the following definition by~\cite{22}.

\begin{definition} \label{truthdef} 

  A quasilinear mechanism is \emph{truthful} if it is direct and
  \(\forall i\), bidder $b_i$'s equilibrium strategy is
  \(\Upsilon(b_i,S_j) = \Pi(b_i,S_j).\)

\end{definition}

(Note that~\cite{22} uses $v_i$ for what we denote by $\Pi(b_i,S_j)$
and $\hat{v_i}$ for what we denote by $\Upsilon(b_i,S_j)$.)

\begin{theorem} \label{truththm}
In our mechanism, the bidder $b_i$'s equilibrium strategy is
  \(\Upsilon(b_i,S_j) = \Pi(b_i,S_j)\) so it is truthful.
\end{theorem}

\begin{proof}
Assume a bidder \(b_i \in B\) who bids for package $S_j$ provides his
fair value $\Pi(b_i,S_j)$ in the fairness table. Let us denote the
strategy choosen by $b_i$. i.e., $\Upsilon(b_i,S_j) = \Pi(b_i,S_j)$ to
be $d_i$ which is a dominant strategy as per our assumption (Line 10 in Algorithm \ref{Algorithm1} and Theorem \ref{Thm2}).

Assume that the bidder $b_i$ would be better off declaring a fair
value $\Pi(b_i,S_j)^\prime$ instead of $\Pi(b_i,S_j)$ to our
mechanism. This implies that $b_i$ has chosen a different strategy
$d_i^\prime$ instead of $d_i$ which is not in equilibrium, contradicting our assumption that $d_i$
is the dominant strategy for $b_i$.\qedhere
\end{proof}

This means that here the only action available to an agent is to
reveal his private information.  Any solution to a mechanism design
problem can be converted into one in which agents always reveal their
true preferences, if the new mechanism ``lies for the agents'' in just
the way they would have chosen to lie to the original mechanism.  Thus
the new mechanism is dominant-strategy truthful~\cite{22}.

In our algorithm the bidder or auctioneer benefit only when they give
their fair value truthfully as in cases where \(P_{i,j} > Q_{a,j}\)
and \(P_{i,j} = Q_{a,j}\), where he gets the incentives as profits are
distributed. But if the fair value is not truthful then he risks going
to the case \(P_{i,j} < Q_{a,j}\) and \(Q_{i,j} < Q_{a,j}\) where
naturally he is denied of any benefits.  Thus if he lies to the
mechanism to gain profits he would not succeed as he would have chosen
a strategy which leads to loss.

\subsubsection{Efficiency} \label{efficiency}

Consider the efficiency with respect to the package won by the bidder
$b_i$ denoted by $S_j$.  We define $S_j^\prime$ as a subset of
resources which are not won by the bidder $b_i$.

Consider the following definition by~\cite{22}.

\begin{definition} \label{effdef}

  A quasilinear mechanism is \emph{strictly Pareto efficient}, or just
  \emph{efficient}, if in equilibrium it selects a choice $S_j$ such
  that \(\sum_i \Pi(S_j) \geq \sum_i \Pi(S_j^\prime)\)

\end{definition}

(Note that~\cite{22} uses $x$ for what we denote by $S_j$.)

\eat{
In the above definition, the choice is represented by $x$. In our
case, the choice is represented by the selection of a package
$S_j$. Hence update the notation for quasilinear utility setting from
$x$ to $S_j$.
}
\begin{theorem} \label{effthm}
In our mechanism, the bidder $b_i$'s equilibrium strategy is to select choice $S_j$ such
  that \(\sum_i \Pi(S_j) \geq \sum_i \Pi(S_j^\prime)\) so it is efficient. 
\end{theorem}

\begin{proof}
Assume a bidder \(b_i \in B\) has chosen a dominant strategy $d_i$
selects a choice $S_j$ such that $\sum_{b_i} \Pi(b_i,S_j) < \sum_{b_i}
\Pi(b_i,S_j^\prime)$. This implies that sum of fair values of items
in selected package $S_j$, is not more efficient than the sum of items
not in package $S_j$. Thus there is another strategy $d_i^\prime$
which selects a choice $S_j^\prime$ which is more efficient than
$S_j$(Theorem \ref{Thm1}). Hence, $d_i$ was not in equilibrium as $d_i^\prime$ is the dominant strategy. This is a
contradiction to our assumption that $d_i$ was the dominant
strategy.\qedhere
\end{proof}

An efficient mechanism selects the choice which maximizes the sum of
the agents' utilities, disregarding the the monetary payments they are
required to pay.  This can be shown in our algorithm concept where the
choice is made on the agents' fair values which helps in maximizing
its profits.  Thus, the efficiency is defined in terms of the true
fair values and not the declared value in the bid table
(Table~\ref{tab:bidtable}).

\subsubsection{Incentive Compatibility} \label{incentive}

The combinatorial auction can be made incentive compatible using the
Generalized Vickrey Auction (GVA) and Algorithm \ref{Algorithm1}.  The
payment using GVA can be explained by assuming that all agents follow
their dominant strategies and declare their values truthfully. Each
agent is made to pay his social cost; the aggregate impact that his
participation has on other agents utilities~\cite{22}.

The payment mechanism described in our system is incentive compatible,
i.e., they fare best when they reveal their private information
truthfully in certain cases. As shown in Theorem \ref{Thm1} and
Theorem \ref{Thm2} the bidders following dominant strategies in
Algorithm \ref{Algorithm1} is bound to get higher incentives.

Thus the GVA and Algorithm \ref{Algorithm1} enables our mechanism to
be incentive compatible.

\subsubsection{Optimality} \label{optimal}

Optimality is a significant property that is desired in a CAS.  We
ensure this property by the use of an algorithm of Sandholm~\cite{18}
in our system.  It is used to obtain the optimum allocation of
resources so as to maximize the revenue generated for the auctioneer.
Thus, the output obtained is the most optimal output and there is no
other allocation that generates more revenues than the current
allocation.

\section{Conclusion} \label{conclusion}

We have shown that fairness can be incorporated in CAS from our
methodology.  Extended fairness as well as basic fairness can be
attained through our payment mechanism.  Optimal allocation is
obtained through an algorithm of Sandholm, and the other significant
properties like allocative efficiency and incentive compatibility are
also achieved.  This is an improvement because in the existing world
of multi-agent systems, there do not seem to be many studies that
attempt to incorporate optimality as well as fairness.  The present
paper addresses this lack in a specific multi-agent system, namely,
the CAS.

The Nigerian Communications Commission (NCC) faced problems in giving
incentives to bidders who divulge their preferences and bidders were
not keen on divulging it since it may lead to more adverse
competition.  Our algorithm for extended fairness takes care of this
problem as bidders receive more incentives with higher bids.  Since
the preferences given by the bidders in the fairness table is
confidential and sealed, they need not worry about their preferences
being disclosed to competitors.

The framework described can also be extended in several ways: first is
to de-centralize the suggested algorithm, to avoid use of a single
dedicated auctioneer. Especially in distributed computing
environments, it would be best for there to be a method to implement
the suggested algorithm (or something close to it) without requiring
an agent to act as a dedicated auctioneer~\cite{30}.

A second important extension would be to find applications for the
work.  Some applications that suggest themselves include distribution
of land (a matter of great concern for governments and people the
world over) in a fair manner.  In land auctions where a tie occurs, no
pre-defined or idiosyncratic method need be used to break the tie;
rather, the allocation can be done fairly in the manner suggested.

A third important extension is to experiment with the grid computing 
framework~\shortcite{35}. The applicability of fairness scheme in grid 
computing while allocating resources and its impact on the expected 
revenue would be an interesting application area.

Fairness is also an important and pressing concern in the computing
sciences and information technology, particularly, in distributed
computing~\cite{12}. It is therefore also of interest to see how our
method for achieving fairness could be applied in such contexts.


\end{document}